\documentclass[11pt, letterpaper,reqno]{amsart}

\usepackage{amsmath}
\usepackage{amssymb}
\usepackage{amsfonts}
\usepackage{ dsfont }
\usepackage{cases}
\usepackage{cite}
\usepackage[hyphens]{url}
\usepackage{hyperref}
\usepackage{amsthm}
\usepackage{upgreek}
\usepackage{color}
\usepackage{xcolor}
\definecolor{color0}{gray}{.50}
\definecolor{color1}{rgb}{0,.2,.8}
\definecolor{color2}{rgb}{1,.2,0}
\definecolor{color3}{rgb}{.8,.5,1}

\numberwithin{equation}{section}

\newcommand{\Max}{{\rm Max\,}}
\newcommand{\Int}{{\rm int\,}}

\newcommand{\cl}{{\rm cl \,}}
\newcommand{\co}{{\rm co \,}}

\newtheorem{theorem}{Theorem}[section]
\newtheorem{corollary}[theorem]{Corollary}
\newtheorem{lemma}[theorem]{Lemma}
\newtheorem{proposition}[theorem]{Proposition}
\newtheorem{definition}[theorem]{Definition}

\newtheorem{remark}[theorem]{Remark}

\newcommand{\leref}{Lemma~\ref}

\newcommand{\coref}{Corollary~\ref}

\newcommand{\prref}{Proposition~\ref}
\newcommand{\thref}{Theorem~\ref}
\newcommand{\reref}{Remark~\ref}

\newcommand{\E}{\mathbb{E}}
\newcommand{\PP}{\mathbb{P}}
\newcommand{\Q}{\mathbb{Q}}
\newcommand{\R}{\mathbb{R}}
\newcommand{\eps}{\epsilon}
\newcommand{\cA}{\mathcal{A}}

\title[]{{Robust no arbitrage and the solvability of vector-valued utility maximization problems}}
\author[]{Andreas H. Hamel}
\address{{Faculty of Economics and Management, Free University of Bozen-Bolzano}}
\email{Andreas.Hamel@unibz.it}
\author[]{Birgit Rudloff} %\thanks{This research was supported in part by the National Science Foundation under grant DMS 0955463.}  
\address{Institute for Statistics and Mathematics, Vienna University of Economics and Business}
\email{birgit.rudloff@wu.ac.at}
\author[]{Zhou Zhou}
\address{School of Mathematics and Statistics, University of Sydney}
\email{zhou.zhou@sydney.edu.au}
\date{\today}
\keywords{Bid-ask spread, robust no arbitrage, Pareto maximizer, utility maximization, consistent price process}

\begin{document}

\begin{abstract}
{A market model with $d$ assets in discrete time is considered where trades are subject to proportional transaction costs given via   bid-ask spreads, while the existence of a num{\'e}raire is not assumed. It is shown that robust no arbitrage holds if, and only if, there exists a Pareto solution for some vector-valued utility maximization problem with component-wise utility functions. Moreover, it is demonstrated that a consistent price process can be constructed from the Pareto maximizer.}
\end{abstract}

\maketitle

\section{Introduction}

The fundamental theorem of asset pricing (FTAP) has been studied extensively in various settings. In the FTAP, the equivalence of no arbitrage (NA) and the existence of an equivalent (super-)martingale measure is established. In market models with transaction costs, the martingale measures are usually replaced by consistent pricing systems (or processes). We refer to \cite{EQF} and the references therein. There is also a stream of literature (e.g., \cite{Sch1, Sch5, Sch6} and the references therein) on general utility maximization with the presence of a num{\'e}raire in the market. In many papers, the existence of solutions (maximizers) is established, and one of the common assumptions is the absence of arbitrage. The easy-to-see fact that the converse is also true is well-known as well, i.e., the existence of a utility maximizer implies NA. For example, in \cite[Section 3.1]{SF}, the authors show both directions for a one-period market model without transaction costs which involves a num{\'e}raire. They show that no arbitrage is equivalent to the existence of a solution for some utility maximization problem. This result goes back to Rogers \cite{Rogers94} who actually constructed equivalent martingale measures from utility maximizers.

The aim of this paper is to provide a corresponding result for market models with proportional transaction costs and a utility maximization problem with component-wise utility functions. This equivalence result between the existence of Pareto maximizers of a such vector-valued utility maximization problem and the well-established robust no-arbitrage condition (introduced in \cite{Sch4}) underlines that from a mathematical point of view the very special case of ``one utility function for each asset" is already enough. Moreover, there is some rational to this situation: the interpretation is that the market model rules the exchange of assets, while ``utilities" (the images of a vector-valued utility function) cannot be exchanged into one another and are therefore compared by means of the order generated by $\R^d_+$ and not a solvency cone (which is usually bigger). It also seems more plausible to have a complete preference for positions of only one asset--instead of a complete preference for portfolio vectors as assumed very often in the literature. Compare \cite{HamelWang17} where this approach was initiated.

A financial market in discrete time is considered in which there are $d$ tradable assets. The existence of a num{\'e}raire is not assumed; instead, following Kabanov \cite{K99}, the assets are considered in physical units where the exchange between these assets are subject to proportional transaction costs. We show that robust no arbitrage holds if, and only if, there exists a solution for some vector-valued utility maximization problem with respect to the order generated by $\R^d_+$ (a Pareto maximizer). A key step to prove the result is to apply Koml{\'o}s's theorem to get the closedness of the set of the vector-valued expected utilities. We also show that a consistent price process can be constructed from the Pareto maximizer.

The note is organized as follows. In the next section, we introduce the setup and the main results. In Section 3, we provide the proofs of the results.

%%%New section
\section{The setup and the main result}
\subsection{Notation and background}

We follow the setup in \cite{Sch4}. Let $(\Omega,(\mathcal{F}_t)_{t=0,\dotso, T},\PP)$ be a filtered probability space in discrete time, where $T\in\mathbb{N}$ is the time horizon, and $\mathcal{F}_0$ is assumed to be trivial. We assume that in the market there are $d$ assets available for exchange. These assets could be interpreted as currencies, but this restriction is not necessary. The exchange of assets is subject to proportional transaction costs; we use it in the form of bid-ask processes as discussed in \cite{Sch4}.

\begin{definition}
A $d\times d$-matrix $\Pi_0=(\pi^{ij})_{i,j=1,\dotso,d}$ is called a bid-ask matrix, if
$$\pi^{ij}>0,\quad\pi^{ii}=1,\quad\text{and}\quad\pi^{ij}\leq\pi^{ik}\pi^{kj},\quad\forall i,j,k=1,\dotso,d.$$
An adapted process $\Pi=(\Pi_t)_{t=0,\dotso,T}$ is called a bid-ask process, if $\Pi_t(\omega)$ is a bid-ask matrix for any $(t,\omega)\in\{0,\dotso,T\}\times\Omega$.
\end{definition}
In the above definition, the entry $\pi_t^{ij}$ in $\Pi_t$ denotes the number of units of asset $i$ an investor needs in order to buy one unit of asset $j$ at time $t$. That is, the bid-ask spread of asset $i$ in terms of asset $j$ is given by $\{1/\pi_t^{ji},\pi_t^{ij}\}$ at time $t$. 

For the bid-ask matrix $\Pi_0=(\pi_{ij})$, let $-K(\Pi_0)$ be the convex cone generated by the vectors $\{-e^i,-\pi^{ij}e^i+e^j,\ i,j=1,\dotso,d\}$, where $e^i=(0,\dotso,0,1,0,\dotso,0)$ with the $i$-th component equal to $1$. In other words, $-K(\Pi_0)$ represents the set of portfolios available starting from initial position $0\in\R^d$ for the assets. Denote $K^*(\Pi_0)$ as the polar of the cone $-K(\Pi_0)$. That is,
$$K^*(\Pi_0):=\{w\in\R^d:\ \langle v,w\rangle\leq 0,\ \forall v\in-K(\Pi_0)\},$$
where $\langle\cdot,\cdot\rangle$ represents the inner product. For any bid-ask process $\Pi=(\Pi_t)_{t=0,\dotso,T}$, similar definitions apply for $-K(\Pi_t(\omega))$ and $K^*(\Pi_t(\omega))$, $(t,\omega)\in\{0,\dotso,T\}\times\Omega$.

Let us recall the definition of self-financing portfolios in \cite{Sch4}. %{\f I assume that $K(\Pi_t)$ as used in the definition below has a parallel definition. Maybe we could define it for any $t$.} 
\begin{definition}
An $\R^d$-valued adapted process $v=(v_t)_{t=0,\dotso,T}$ is called a self-financing portfolio process (from initial position $0$) for the bid-ask process $\Pi=(\Pi_t)_{t=0,\dotso,T}$, if for $t=0,\dotso,T$,
$$v_t-v_{t-1}\in-K(\Pi_t),\quad\text{ a.s.},$$
where $v_{-1}\equiv 0$.
Denote
$$A_\Pi:=\{v_T:\ \text{for some self-financing portfolio process }(v_t)_{t=0,\dotso,T}\text{ in terms of }\Pi\}.$$
\end{definition}

Recall the notion of (robust) no arbitrage introduced in \cite{Sch4}.

\begin{definition}
We say NA holds with respect to the bid-ask process $\Pi$, if for any $v_T\in A_\Pi$, if $v_T\geq 0$ a.s. then $v_T=0$ a.s., where the (in)equalities are component-wise. We say NA$^r$ holds with respect to $\Pi$, if there exists a bid-ask process $\tilde\Pi=(\tilde\Pi_t)_t$ with smaller bid-ask spreads (i.e., for $t=0,\dotso,T$ and $i,j = 1,\dotso,d$, $[1/\tilde\pi_t^{ji},\tilde\pi_t^{ij}]$ is contained in the relative interior of $[1/\pi_t^{ji},\pi_t^{ij}]$ a.s.), such that NA holds with respect to $\tilde\Pi$.
\end{definition}

Let us also recall the definition of (strictly) consistent price processes.
\begin{definition}
	An $\R_+^d$ valued adapted process $Z=(Z_t)_{t=0,\dotso,T}$ is called a consistent (strictly consistent, respectively) price process for the bid-ask process $\Pi$, if $Z$ is a $\PP$-martingale, and $Z_t(\cdot)\in K_t^*(\cdot)\setminus\{0\}$ (the relative interior of $K_t^*(\cdot)$, respectively) a.s., $t=0,\dotso,T$.
\end{definition}

For $i=1,\dotso,d$, let $U^i:\R_+\mapsto[-\infty,\infty)$ be \textit{concave, continuous, strictly increasing, and bounded from above}, representing the utility function for asset $i$. Define $U:\R_+^d\mapsto\R^d$,
\[
U(x)=(U^1(x^1),\dotso,U^d(x^d)),\quad x=(x^1,\dotso,x^d)\in\R_+^d.
\]
For a given initial endowment $x\in\R_+^d$, consider the $\R^d$-valued utility maximization problem
\begin{equation}\label{e1}
\text{maximize} \quad \E_\PP U(X_T) \quad \text{subject to} \;  X_T\in\cA_\Pi(x),
\end{equation}
where
\[
\cA_\Pi(x):=\{X_T\geq 0\text{ a.s.}:\ X_T\in A_\Pi+x\}
\]
and
\[
\E_\PP U(X_T):=(\E_\PP U^1(X_T^1),\dotso,\E_\PP U^d(X_T^d)),\quad X_T=(X_T^1,\dotso,X_T^d)\in\cA_\Pi(x).
\]
Moreover, we also consider the set 
\[
J_\Pi(U)(x) := \{\E_\PP U(X_T) \colon X_T\in A_\Pi(x)\}.
\]

The meaning of the word ``maximize" in \eqref{e1} has to be defined. We will consider two variants. First, we recall the definition of a Pareto maximizer with respect to the ordering cone $\R^d_+$ and a domination property which is also called (upper) external stability in the literature. For yet another notion, compare \reref{r1} below.

\begin{definition}
(a) A point $\hat z \in M$ is called Pareto maximal for $M \subseteq \R^d$ iff
\[
M \cap (\hat z + \R^d_+) = \{\hat z\}.
\]
The set of Pareto maximal points of $M$ is denoted by $\Max M$.

(b) A set $M \subseteq \R^d$ is said to satisfy the upper domination property if for each $z \in M$ there is $\hat z \in \Max M$ such that $z \leq_{\R^d_+} \hat z$.
\end{definition}

\begin{definition}
A position $\hat X_T \in \cA_\Pi(x)$ is called a Pareto maximizer for \eqref{e1} iff $\E_\PP U(\hat X_T) \in \Max J_\Pi(U)(x)$. The  set of Pareto maximizers for \eqref{e1} is denoted by $P_\Pi(U)(x)$.
\end{definition}

One may easily realize that $\hat X_T \in P_\Pi(U)(x)$ if, and only if, 
there does not exist $X_T\in\cA_\Pi(x)$, such that 
\begin{itemize}
\item[i)] For all $i\in\{1,\dotso,d\},\ \E_\PP U^i(X_T^i)\geq\E_\PP U^i(\hat X_T^i)$;
\item[ii)] For some $i\in\{1,\dotso,d\},\ \E_\PP U^i(X_T^i)>\E_\PP U^i(\hat X_T^i)$.
\end{itemize}

The study of the vector-valued utility maximization problem \eqref{e1} was initiated in \cite{Wang10} and \cite{Wang11}, the former containing a certainty equivalent, the latter duality results in a one-period setup with a finite probability space which were published in \cite{HamelWang17}. In \cite{Zamboni15}, the problem is studied with each of the components of the utility function depending on all the components of the commodity vector in a deterministic and set-valued framework. In \cite{2019arXiv190409456R}, further certainty equivalents as well as indifference prices are defined also for a finite probability space.

%%%New subsection
\subsection{Main results}

The following two results, \prref{t1} and \thref{t2}, establish the equivalence between no arbitrage and the existence of a Pareto optimizer for \eqref{e1}.

\begin{proposition}\label{t1}
Let $\Pi$ be a bid-ask process. If NA$^r$ holds with respect to $\Pi$, then $J_\Pi(U)(x)$ satisfies the upper domination property. In particular, $P_\Pi(U) \neq \emptyset$. 

Conversely, if $P_\Pi(U) \neq \emptyset$, then NA holds with respect to $\Pi$.

%{\f
%(2) $P_\Pi(U) \neq \emptyset$  for $U$ with exponential components

%(3) The set $J:=\{\E_\PP U(X_T) \colon X_T \in \cA_\Pi(x)\} \subset\R^d$ has the {\em upper domination property}, i.e. for all $z \in J$ there is $\bar z \in P_\Pi$ such that $\bar z - z \in \R^d_+$ for each $U$ with concave, strictly increasing, bounded components

%(4) The $\mathcal G(\R^d, -\R^d_+)$-valued extension of problem \eqref{e1} has a full solution for each $U$ with concave, strictly increasing, bounded components

%(5)  The $\mathcal G(\R^d, -\R^d_+)$-valued extension of problem \eqref{e1} has a full solution for each $U$ with exponential components
%}
\end{proposition}

\begin{theorem}\label{t2}
NA$^r$ holds with respect to the bid-ask process $\Pi$ if, and only if, there exists a bid-ask process $\tilde\Pi$ with smaller bid-ask spreads satisfying $P_{\tilde\Pi}(U) \neq \emptyset$.
\end{theorem}
%
%{\f (1) and (2) are covered by your proof, Zhou, or easy. (3) is your proof with $J$ replaced by $J \cap (z + \R^d_+)$ (this is what needs to be compact, I guess) for $z \in J$. (3) $\Rightarrow$ (4) follows from general results in set optimization (L\"ohne's book, I'll provide a reference.) and the converse is by definition of a {\em full solution}. (5) could probably be involved in the same way as (2).}

%{\f \begin{remark}
%If $U$ has exponential components, then the solution of problem \eqref{e1} provides a strictly consistent price process by means of the device in Corollary 3.10 of F\"ollmer/Schied, 2nd edition. Indeed, to be added.
%\end{remark}}

\begin{remark}\label{r1}
A comment on the above results might be in order. The function $X_T \mapsto \E_\PP U(X_T)$ can straightforwardly be extended to a function $F$ mapping into the set $\mathcal G(\R^d, -\R^d_+) := \{D \subseteq \R^d \colon D = \cl\co(D - \R^d_+)\}$ by setting
\[
F(X_T) = \E_\PP U(X_T) - \R^d_+.
\]
The pair $(\mathcal G(\R^d, -\R^d_+), \subseteq)$ is a complete lattice. The first part of \prref{t1} can now be read as follows: NA$^r$ implies the existence of a solution of the set optimization problem (in the complete-lattice sense, see, e.g., \cite[Definition 28, Corollary 2.9]{Loehne11}) 
\[
\text{maximize} \quad F(X_T) \quad \text{subject to} \;  X_T\in\cA_\Pi(x).
\]
This follows from \cite[Proposition 2.15]{Loehne11} (adapted to maximization), and in this reference more material about the treatment of vector optimization as complete lattice-valued problems can be found.

Thus, roughly speaking, \thref{t2} yields that robust no arbitrage is equivalent to the existence of solutions of the complete lattice extension of some vector-valued utility maximization problem. One should note that the existence of such a solution does not imply the domination property.
\end{remark}

Motivated by \cite[Corollary 3.10]{SF}, we also establish the following wo results which provide a construction of a (strictly) consistent price process from a Pareto maximizer.

\begin{proposition}\label{p1}
Let NA$^r$ hold with respect to the bid-ask process $\Pi$ and let $\hat X_T\in P_\Pi(U)$. Assume that $U^i$ is differentiable for all $i=1,\dotso,d$, and that there exists $\delta \in\Int \R^d_+$ such that $\hat X_T \geq \delta$ a.s. Then, there exists $(\lambda^1,\dotso,\lambda^d)\in\R_+^d\backslash\{0\}$ such that
$$\left(\left(\lambda^i\E_\PP[(U^i)'(\hat X_T^i)|\mathcal{F}_t]\right)_{i=1,\dotso,d}\right)_{t=0,\dotso,T}$$
defines a consistent price process for $\Pi$.
\end{proposition}

\begin{corollary}\label{c1}
Let NA$^r$ hold with respect to the bid-ask process $\Pi$, and thus NA$^r$ still holds with respect to some $\tilde\Pi$ with smaller bid-ask spreads. Let $\hat X_T\in P_{\tilde\Pi}(U)$. Assume that $U^i$ is differentiable for all $i=1,\dotso,d$, and that there exists $\delta\in\Int \R^d_+$ such that $\hat X_T \geq \delta$ a.s. Then, there exists $(\lambda^1,\dotso,\lambda^d)\in\Int \R^d_+$ such that
$$\left(\left(\lambda^i\E_\PP[(U^i)'(\hat X_T^i)|\mathcal{F}_t]\right)_{i=1,\dotso,d}\right)_{t=0,\dotso,T}$$
defines a strictly consistent price process for $\Pi$.
\end{corollary}
%
%{\f \begin{remark}
%Of course, Pareto points in dimensions $\geq 2$ have very little chance of being unique. However, if $U$ has strictly concave components, then the full solution in (5) of the above theorem is unique in the following sense: to be added.
%\end{remark}}
%
%

\section{Proof of the results}

\begin{lemma}\label{l1}
Assume NA$^r$ holds with respect to $\Pi$. Then for any sequence $(X_T^n)_{n\in\mathbb{N}}\subset\cA_\Pi(x)$, there exists a subsequence $(X_T^{n_k})_{k\in\mathbb{N}}\subset(X_T^n)_{n\in\mathbb{N}}$ and $Y\in\cA_\Pi(x)$, such that
$$\frac{1}{N}\sum_{k=1}^N X_T^{n_k}\rightarrow Y\ \text{a.s.},\quad N\rightarrow\infty.$$
\end{lemma}

\begin{proof}
By \cite[Theorem 1.7]{Sch4}, there exists a strictly consistent price process $(Z_t)_{t=0,\dotso,T}$. For $i=1,\dotso,d$, define
$$\frac{d\Q^i}{d\PP}:=\frac{Z_T^i}{Z_0^i}.$$
Then, $\Q^i$ is a probability measure that is equivalent to $\PP$, $i=1,\dotso,d$. For any $X_T=(X_T^1,\dotso,X_T^d)\in\cA_\Pi(x)$, we have that for $i=1,\dotso,d$,
\begin{equation}\label{e2}
\E_{\Q^i} X_T^i=\frac{1}{Z_0^i}\E_\PP[Z_T^i X_T^i]\leq\frac{1}{Z_0^i}\E_\PP\langle Z_T,X_T\rangle\leq\frac{1}{Z_0^i}\langle Z_0,x\rangle,
\end{equation}
where the second (in)equality follows from the fact that the components of $Z_T$ and $X_T$ are all nonnegative a.s., and the third (in)equality follows from \cite[Theorem 4.1]{Sch4}.

Now, let $(X_T^n)_{n\in\mathbb{N}}\subset\cA_\Pi(x)$, where $X_T^n=(X_T^{n,1},\dotso,X_T^{n,d})$. First, consider the first component, i.e., the sequence $ (X_T^{n,1})_{n\in\mathbb{N}}$. Since from \eqref{e2}
$$\sup_{n\in\mathbb{N}}\E_\Q X_T^{n,1}\leq\frac{1}{Z_0^1}\langle Z_0,x\rangle<\infty,$$
by Koml{\'o}s's theorem (see e.g., \cite[Theorem 1a]{Komlos}), there exists a subsequence $(X_T^{n_k,1})_{k\in\mathbb{N}}\subset(X_T^{n,1})_{n\in\mathbb{N}}$ and a random variable $Y^1\in\mathcal{F}_T$, such that for any further subsequence $(\eta^j)_{j\in\mathbb{N}}\subset(X_T^{n_k,1})_{k\in\mathbb{N}}$,
\begin{equation}\label{e3}
\frac{1}{N}\sum_{j=1}^N \eta^j\rightarrow Y^1\ \text{a.s.},\quad N\rightarrow\infty.
\end{equation}
Next, consider the sequence $(X_T^{n_k,2})_{k\in\mathbb{N}}\subset(X_T^{n,2})_{n\in\mathbb{N}}$ for the second component. Applying Koml{\'o}s's theorem again, we have that there exist a subsequence $(X_T^{n_{k_l},2})_{l\in\mathbb{N}}\subset(X_T^{n_k,2})_{k\in\mathbb{N}}$ and $Y^2\in\mathcal{F}_T$, such that for any further subsequence $(\xi_j)_{j\in\mathbb{N}}\subset(X_T^{n_{k_l},2})_{l\in\mathbb{N}}$,
$$\frac{1}{N}\sum_{j=1}^N \xi^j\rightarrow Y^2\ \text{a.s.},\quad N\rightarrow\infty.$$
Obviously, \eqref{e3} still holds if we replace the $(X_T^{n_k,1})_{k\in\mathbb{N}}$ with $(X_T^{n_{k_l},1})_{l\in\mathbb{N}}$. Repeating this process for the remaining components, we can finally get a subsequence $(\zeta_i)_{i\in\mathbb{N}}\subset(X_T^n)_{n\in\mathbb{N}}$ and an $\R^d$-valued random variable $Y=(Y^1,\dotso,Y^d)\in\mathcal{F}_T$, such that
$$\frac{1}{N}\sum_{i=1}^N\zeta_i\rightarrow Y\ \text{a.s.},\quad N\rightarrow\infty.$$
Since $\cA_\Pi(x)$ is convex, $\frac{1}{N}\sum_{i=1}^N\zeta_i\in\cA_\Pi(x)$. By \cite[Theorem 2.1]{Sch4}, $\cA_\Pi(x)$ is closed in $L^0(\Omega,\mathcal{F}_T,\PP;\R^d)$, and thus $Y\in\cA_\Pi(x)$.
\end{proof}

\begin{proof}[\textbf{Proof of \prref{t1}}]
First, let us show the first statement of this proposition. Assume that NA$^r$ holds with respect to $\Pi$. It suffices to show that the set $J_\Pi(U)(x) \subset \R^d$ is closed. Indeed, if this is the case, then the set $(\E_\PP U(X_T) + \R_+^d) \cap J_\Pi(U)(x)$ is compact for each $X_T \in \cA_\Pi(x)$ since $U$ is component-wise bounded from above. Therefore, by \cite[Theorem 2(i), page 489]{Luc08}, there is a Pareto maximizer $\bar X_T \in \cA_\Pi(x)$ satisfying $\E_\PP U(X_T) \leq_{\R^d_+} \E_\PP U(\bar X_T)$. In particular, $P_\Pi(U) \neq \emptyset$.

In order to show the closedness of $J_\Pi(U)(x)$, take $(X_T^n)_{n\in\mathbb{N}}\subset\cA_\Pi(x)$, such that
$$\E_\PP U(X_T^n)\rightarrow a\in\left(\prod_{i=1}^d[U^i(0),\infty)\right)\cap\R^d.$$
By \leref{l1}, there exists a subsequence $(X_T^{n_k})_{k\in\mathbb{N}}$ and $Y\in\cA_\Pi(x)$, such that
$$\frac{1}{N}\sum_{k=1}^N X_T^{n_k}\rightarrow Y\ \text{a.s.},\quad N\rightarrow\infty.$$
As $U_i$ is concave, $i=1,\dotso,d$, we have that
$$\E_\PP U\left(\frac{1}{N}\sum_{k=1}^N X_T^{n_k}\right)\geq\frac{1}{N}\sum_{k=1}^N \E_\PP U(X_T^{n_k}),$$
where the inequality is component-wise. As $U$ is continuous and bounded from above, by Fatou's Lemma, one has
$$\E_\PP U(Y)\geq a.$$
Because $\cA_\Pi(x)=(\cA_\Pi(x)-L^0(\Omega,\mathcal{F}_T,\PP;\R_+^d)\cap L^0(\Omega,\mathcal{F}_T,\PP;\R_+^d)$ and $U^i$ is continuous for $i=1,\dotso,d$,  there exists $\tilde Y\in\cA_\Pi(x)$, such that $\E_\PP U(\tilde Y)=a$. This completes the proof of the first statement.

Next, let us prove the second statement. Let $\hat X_T\in\cA_\Pi(x)$ be a Pareto maximizer for \eqref{e1}. If NA fails with respect to $\Pi$, then there would exist some $\bar X_T=(\bar X_T^1,\dotso,\bar X_T^d)\in A_\Pi$, such that
$$\text{for any }i\in\{1,\dotso,d\},\quad\bar X_T^i\geq 0\ \text{a.s.},$$
and
$$\text{for some }i\in\{1,\dotso,d\},\quad\PP\{\bar X_T^i>0\}>0.$$
Then $\hat X_T+\bar X_T\in\cA_\Pi(x)$, and
$$\text{for any }i\in\{1,\dotso,d\},\quad\E_\PP U^i(\hat X_T^i+\bar X_T^i)\geq \E_\PP U^i(\hat X_T^i),$$
and
$$\text{for some }i\in\{1,\dotso,d\},\quad\E_\PP U^i(\hat X_T^i+\bar X_T^i)>\E_\PP U^i(\hat X_T^i).$$
This contradicts the Pareto optimality of $\hat X_T$.
\end{proof}

\begin{proof}[\textbf{Proof of \thref{t2}}]
If NA$^r$ holds with respect to $\Pi$, then there exists $\tilde\Pi$ with smaller bid-ask spreads, such that NA$^r$ also holds with respect to $\tilde\Pi$. By \prref{t1} one has $ P_{\tilde\Pi}(U) \neq \emptyset$.

Conversely, if $P_{\tilde\Pi}(U) \neq \emptyset$ for some $\tilde\Pi$ with smaller bid-ask spreads, then by \prref{t1} NA holds with respect to $\tilde\Pi$, and thus NA$^r$ holds with respect to $\Pi$.
\end{proof}

\begin{proof}[\textbf{Proof of \prref{p1}}]
Since $\hat X_T$ is a Pareto optimizer for \eqref{e1}, by \cite[Proposition 9, page 497]{Luc08} there exists $\lambda=(\lambda^1,\dotso,\lambda^d)\in\R_+^d\backslash\{0\}$, such that $\hat X_T$ is also optimal for the scalar-valued maximization problem
\begin{equation}\label{e4}
\sup_{X_T\in\cA_\Pi(x)}\langle\lambda,\E_\PP U(X_T)\rangle.
\end{equation}
Fix $t\in\{0,\dotso,T\}$. Let $u_t$ be an $\R^d$-valued $\mathcal{F}_t$-measurable random variable satisfying $u_t(\cdot)\in-K(\Pi_t(\cdot))$ and $u_t(\cdot)\geq-\delta/2$ a.s. Since $\hat X_T+\eps u_t\in\cA_\Pi(x)$ for any $\eps\in(0,1)$ and $\hat X_T$ is optimal for \eqref{e4}, we have that 
$$\langle\lambda,\E_\PP U(\hat X_T+\eps u_t)\rangle-\langle\lambda,\E_\PP U(\hat X_T)\rangle\leq 0.$$
Then
$$\sum_{i=1}^d\lambda_i\E_\PP\left[\frac{U^i(\hat X_T^i+\eps u_t^i)-U^i(\hat X_T^i)}{\eps}\right]\leq0.$$
Letting $\eps\searrow 0$ and applying the dominated convergence theorem, we get that
$$\sum_{i=1}^d\lambda_i\E_\PP\left[u_t^i\cdot(U^i)'(\hat X_T^i)\right]\leq0.$$
Therefore,
\begin{equation}\label{e5}
\E_\PP\left[\sum_{i=1}^d\lambda_i\E_\PP\left[(U^i)'(\hat X_T^i)|\mathcal{F}_t\right]u_t^i\right]\leq0.
\end{equation}

Let us show that
$$\left(\lambda_1\E_\PP\left[(U^1)'(\hat X_T^1)|\mathcal{F}_t\right],\dotso, \lambda_d\E_\PP\left[(U^d)'(\hat X_T^d)|\mathcal{F}_t\right]\right)(\cdot)\in K^*(\Pi_t(\cdot)),\quad\text{a.s.}$$
Indeed, if the above fails to be true, then there exists an $\mathcal{F}_t$-measurable $\bar w_t(\cdot)\in -K(\Pi_t(\cdot))$ such that
$$\PP\left\{\sum_{i=1}^d\left(\lambda_i\E_\PP\left[(U^i)'(\hat X_T^i)|\mathcal{F}_t\right]\cdot \bar w_t^i\right)>0\right\}>0.$$
Then there exists $c>0$ such that
$$\PP\left\{c\bar w_t>-\delta/2,\ \sum_{i=1}^d\left(\lambda_i\E_\PP\left[(U^i)'(\hat X_T^i)|\mathcal{F}_t\right]\cdot c \bar w_t^i\right)>0\right\}>0.$$
Denote the set in the above as $B$. If we replace $u_t$ in \eqref{e5} by $\tilde w_t:=c\bar w_t1_B$, then the inequality \eqref{e5} fails, which yields a contradiction.
\end{proof}

\begin{proof}[\textbf{Proof of \coref{c1}}]
The result follows from \prref{p1} and  \cite[Proposition A.5]{Sch4}.
\end{proof}

{\bf Contribution of the authors.} 
The first author conjectured the results of this paper (in form of Proposition \ref{t1} and Remark \ref{r1}); this conjecture is already mentioned in \cite{Wang10, Wang11} supervised by the first author regarding the content. The equivalence problem for solutions of vector-utility functions and (robust) no arbitrage for markets with transaction costs was among the open problems given by the first author to the second for discussion at the AMS Mathematics Research Community 2015 in Financial Mathematics. The collaboration of the second and the third author at this summer school led to the proof of this conjecture. 

{\bf Acknowledgement.} 
The last two authors gratefully acknowledge the financial support by the National Science Foundation, Division of Mathematical Sciences, under Grant No. 1321794.

\bibliographystyle{siam}

\end{document}